\documentclass[journal]{IEEEtran} 

\usepackage{mathrsfs}
\usepackage{amsfonts}
\usepackage{amsmath,amssymb}
\usepackage{graphicx}
\usepackage{float}
\usepackage[dvips]{epsfig}
\usepackage[dvips]{color}
\usepackage{subfigure}
\usepackage{algorithm}
\usepackage{algorithmic}

\def\dref#1{(\ref{#1})}

\newtheorem{assumption}{Assumption}
\newtheorem{lemma}{Lemma}
\newtheorem{theorem}{Theorem}

\newtheorem{remark}{Remark}
\newtheorem{definition}{Definition}
\newtheorem{proof}{Proof}

\begin{document}

%
\title{Privacy-Preserved Average Consensus Algorithms with Edge-based Additive Perturbations}

\author{Yi Xiong and Zhongkui Li
\thanks{This work was supported by the National Natural Science Foundation of China under grants 61973006 and U1713223.}
\thanks{Y. Xiong and Z. Li are with College of Engineering, Peking University, Beijing 100871, China. E-mail: {\tt xiongyi@pku.edu.cn, zhongkli@pku.edu.cn}}
}

\IEEEtitleabstractindextext{%
\begin{abstract}
In this paper, we consider the privacy preservation problem in both discrete- and continuous-time average consensus algorithms with strongly connected and balanced graphs, against either internal honest-but-curious agents or external eavesdroppers. A novel algorithm is proposed, which adds edge-based perturbation signals to the process of consensus computation. Our algorithm can be divided into two phases: a coordinated scrambling phase, which is for privacy preservation, and a convergence phase. In the scrambling phase, each agent is required to generate some perturbation signals and add them to the edges leading out of it. In the convergence phase, the agents update their states following a normal updating rule. It is shown that an internal honest-but-curious agent can obtain the privacy of a target agent if and only if no other agents can communicate with the target agent. As for external eavesdroppers, it is proved that this kind of attackers can never obtain any agent's privacy.
\end{abstract}

}

\maketitle

\IEEEdisplaynontitleabstractindextext

\IEEEpeerreviewmaketitle

\section{Introduction}
In the last decade, preserving privacy in average consensus has received increasing attentions in the systems and control community. Traditional consensus algorithms, e.g., those in \cite{OlfatiSwitching,Boyd,li2011consensus}, require
agents to directly exchange their states information for computation. This brings some security issues in the sense that the agents' initial states would be easily revealed to their out-neighbors and possible external adversaries. The disclosure of agents' initial states may be undesirable, since they usually contain some important and sensitive information. An example is that some participants of an organization want to reach a common opinion with a consensus algorithm but they do not want their own opinions known by others \cite{Mo2016Privacy}. Therefore, as long as cyber security is concerned,  consensus algorithms should be secure and have the ability of preventing internal or external adversaries from obtaining agents' privacy, i.e., their initial states.

The privacy preserving average consensus problem has been extensively investigated by many researchers. A natural idea is to encrypt the transmitted messages so that they will not be accessed by the public or third parties \cite{WangPushSum,YuezuLv2019accurate}. By making use of the Pailler cryptosystem and designing an ingenious communication mechanism, privacy preservation was successfully embedded in pairwise communication in \cite{WangCryptography}. The ratio consensus and homomorphic cryptosystems are combined in \cite{hadjicostis2020privacy}  to preserve privacy. Nevertheless, cryptology-based methods share a common weakness that they usually put heavy burden on the communication and computation on the network.

Many works, e.g., \cite{Huang,Mo2016Privacy,Manitara2013Privacy,Kia,He2018privacy,He2018preserving,KiaArxiv,GuptaPrivacy}, protect privacy by adding random noises or deterministic perturbation signals to the transmitted messages. Some of them utilized a tool called differential privacy \cite{cortes2016differential,DifferentialPrivacy_0}. For instance, \cite{Huang} and \cite{diffPrivacyOpt} added random noises according to a Laplace distribution to messages in order to protect the real initial states. It should be pointed out that in the existing works based on differential privacy, accurate average consensus cannot be reached. To overcome this drawback, some researchers designed correlated noises ingeniously. In \cite{Mo2016Privacy} and \cite{He2018privacy}, a sequence of zero-sum correlated noises were added to the messages to preserve privacy. Nevertheless, these correlated-noise-based works require undirected graphs and does not consider external adversaries. Instead of adding random noises, the authors in \cite{Kia}  inserted some well-designed deterministic perturbation signals into the classic consensus algorithm. In \cite{GuptaPrivacy}, the authors proposed an interesting privacy-preserving consensus algorithm, which adds edge-based perturbations. There are also some other methods to preserve privacy, such as the state decomposition \cite{WangStateDecomposition}, output masking \cite{dyMask}, hot-pluggable methods \cite{hotPluggable} and observability based methods \cite{dyCon,netDesign,adpNetPrivacy}.

In this paper, we intend to propose a novel privacy preserving consensus algorithm, which can achieve accurate average consensus for a broader class of network graphs and under a much relaxed privacy-preserving condition, compared to the existing literature. It is observed that in those aforementioned works such as \cite{Mo2016Privacy,He2018privacy,Kia}, each agent adds noises or perturbation signals to its original state and then broadcasts the obfuscated value to all its neighboring agents. The privacy preserving performance, nevertheless, might be compromised, since the heterogeneity of the network in the sense that the agents usually have different amounts of neighbors is not fully exploited. Motivated by this observation, in this paper we propose to insert some well-designed perturbation signals to the communication edges in the network. This method of adding edge-based perturbation signals, different from the node-based methods in the aforementioned works, brings more degrees of design freedom.

Specifically, in this paper we consider the privacy preservation problem in both discrete- and continuous-time average consensus algorithms with strongly connected and balanced graphs, against either internal or external attackers. By the proposed edge-based method, each agent designs a different perturbation signal for every edge leading out of it and then sends the corrupted values to its neighbors. A distinct feature is that perturbation signals does not to be added all the time. In the discrete-time case, each agent only adds perturbation signals at the first iteration, while in the continuous-time case the perturbation signals are injected only in a finite time interval. It is shown that an internal honest-but-curious agent can obtain privacy of a target agent if and only if no other agents except the internal attacker can communicate with the target agent. It is further proved that  external eavesdroppers cannot obtain any agent's privacy.

Compared to the existing privacy preserving average consensus algorithms, the algorithm proposed in this paper have several advantages. In contrast to the differential-privacy-based methods in \cite{Huang,diffPrivacyOpt}, our method can achieve accurate average consensus. Besides, our algorithm is valid when the graph is a balanced directed graph, while most existing works, e.g., \cite{WangCryptography,Mo2016Privacy,He2018privacy,WangStateDecomposition}, require undirected graphs. Compared to those in \cite{Kia}, our algorithm might perform better against the internal adversaries. In \cite{Kia}, an internal attacker can obtain an agent's privacy if and only if the agent itself and all its in-neighbors are the attacker's in-neighbors. By contrast, in the current paper the privacy disclosure happens only when the internal attacker is the only agent that can communicate with the target agent. Note that \cite{GuptaPrivacy} also proposed an edge-based privacy-preserving consensus algorithm, which requires undirected graphs and considers only internal attackers.  The algorithm in this paper, nevertheless, is simpler and is applicable to the case when the graph is directed and balanced and when external attackers exist. Moreover, the algorithm design and privacy analysis in this paper are based on a system-theorectic framework, significantly different form those in \cite{GuptaPrivacy} which are from a perspective of probability analysis.

The rest of this paper is organized as follows. The privacy preservation problem is formulated in Section \uppercase\expandafter{\romannumeral2}. A discrete-time privacy preserving average consensus algorithm is proposed and its properties are analyzed in Section \uppercase\expandafter{\romannumeral3}. Extensions to the continuous-time case are considered in Section \uppercase\expandafter{\romannumeral4}. Numerical simulations are conducted in Section \uppercase\expandafter{\romannumeral5} to demonstrate the effectiveness of our discrete-time algorithm. Finally, this paper is concluded in Section \uppercase\expandafter{\romannumeral6}.

\section{Problem Formulation}\label{background}
In this paper, we consider a network of $N$ ($N\geq 3$) agents. The communication among agents is modeled by a graph $\mathcal{G}$,  consisting of a node set $\mathcal{V}$ and an edge set $\mathcal{E}\subseteq\mathcal{V}\times\mathcal{V}$.
If $(i,j)\in\mathcal{E}$, it means that agent $i$ can send information to agent $j$, in which case agent $i$ is said to be an in-neighbor of agent $j$ and agent $j$ is said to be an out-neighbor of agent $i$.
A graph with the property that $d_i^{in}=d_i^{out},\mbox{ }\forall i\in\mathcal{V}$, is said to be balanced, where $d_i^{in}$ (respectively, $d_i^{out}$) denotes the cardinality of agent $i$'s in-neighbor set $\mathcal{N}_i^{in}$ (respectively, out-neighbor set $\mathcal{N}_i^{out}$). In such a graph, define each node's degree as $d_i=d_i^{in}=d_i^{out}$.

\begin{assumption}
The graph $\mathcal{G}$ is strongly connected and balanced.
\end{assumption}

Two types of attackers are considered in this paper. The first type is called internal honest-but-curious agents. This type of attackers are some agents in the network, which implement our algorithm correctly but will try to evaluate other agents' initial states via the collected information. The second type is called external eavesdroppers. This type of attackers exist outside the network and does not take part in the process of consensus computation, but they are able to hack into communication links and have access to all transmitted information among neighboring agents.

\begin{lemma}[\cite{OlfatiSwitching,bookLi2014Cooperative}]\label{lemma_normal_rule_discrete}
For a strongly connected and balanced graph, the following updating rule guarantees accurate average consensus:
\begin{equation}\label{Normal_update_rule_discrete}
x_i[k+1]=x_i[k]+\epsilon\sum_{j=1}^N a_{ij}(x_j[k]-x_i[k]),\mbox{ }\forall k\geq 0,
\end{equation}
where $\epsilon\in(0,\frac{1}{max\{d_1,\cdots,d_N\}})$ is a constant and $a_{ij}$ denotes the $(i,j)$-th entry of the adjacency matrix $\mathbf{A}$, defined as $a_{ii}=0$, $a_{ij}=1$ if $(j,i)\in\mathcal{E}$ and $a_{ij}=0$ otherwise.
\end{lemma}

\begin{lemma}[\cite{OlfatiSwitching,bookLi2014Cooperative}]\label{lemma_normal_rule_continuous}
For a strongly connected and balanced graph, the average consensus is achieved under the following continuous-time updating rule:
\begin{equation}\label{Normal_update_rule_continuous}
\dot{x}_i(t)=-c\sum_{j=1}^Na_{ij}(x_i(t)-x_j(t)),\mbox{ }\forall t\geq 0,
\end{equation}
where $c>0$ is a constant.
\end{lemma}

The consensus algorithms \dref{lemma_normal_rule_discrete} and \dref{Normal_update_rule_continuous} cannot preserve privacy. For example, it is obvious in \dref{Normal_update_rule_discrete} that all agents' initial states information will be transmitted at the first iteration, implying that an agent's privacy will be revealed to its out-neighbors and all agents' privacy will be revealed to the external eavesdroppers.

The goal of this paper is to design a novel consensus algorithm that can fulfill two tasks: \romannumeral1) to achieve accurate average consensus, and \romannumeral2) to prevent the aforementioned two types of attackers from obtaining agents' privacy.

\section{Discrete-time Privacy Preserving Updating Rule and Privacy Analysis}
In this section, we consider the discrete-time consensus case.

The proposed algorithm is stated as follows. At the first iteration, each agent generates some perturbation signals and then add them to the edges leading out of it. Specially, for each agent $i$, if $(i,j)\in\mathcal{E}$, then agent $i$ choose a real number $p_i^{(j)}\in\mathbb{R}$ and transmit $y_i^{(j)}[0]=x_i[0]+p_i^{(j)}$ to its out-neighbor, agent $j$. All the agents then update their states as
\begin{equation}\label{iter_0_discrete}
x_i[1]=x_i[0]+\epsilon_1\sum_{j=1}^N a_{ij} (y_j^{(i)}[0]-x_i[0])+\epsilon_1 p_i^{(i)},
\end{equation}
where $\epsilon_1\in(-\infty,0)\bigcup(0,\infty)$ is a constant and $p_i^{(i)}=-\sum_{j\in\mathcal{N}_i^{out}}p_i^{(j)}$. For $k\geq 1$, all agents transmit their true state values and implement the algorithm in Lemma \ref{lemma_normal_rule_discrete}, i.e.,
\begin{equation}\label{iter_from_1_discrete}
x_i[k+1]=x_i[k]+\epsilon_2\sum_{j=1}^N a_{ij}(x_j[k]-x_i[k]), \mbox{ }\forall k\geq 1,
\end{equation}
where $\epsilon_2\in(0,\frac{1}{max\{d_1,\cdots,d_N\}})$ is a constant.

The algorithm is summarized in Algorithm 1.

\begin{algorithm}[htb]
\caption{ Privacy Preserving Average Consensus Algorithm.}
\label{alg:our_alg_discrete}
\hspace*{0.02in} {\bf When $k=0$:}
\begin{algorithmic}[1] 
\STATE Each agent $i$ chooses $d_i$ numbers from the set of real numbers $\mathbb{R}$: $p_i^{(j_1)},\cdots,p_i^{(j_{d_i})}$, where $j_1,\cdots,j_{d_i}\in\mathcal{N}_i^{out}$;
\STATE Each agent $i$ computes $p_i^{(i)}=-\sum_{j\in\mathcal{N}_i^{out}}p_i^{(j)}$;
\STATE Agent $i$ transmits $y_i^{(j)}[0]$ to agent $j$, $\forall j\in\mathcal{N}_i^{out}$;
\STATE Each agent $i$ updates its state according to \dref{iter_0_discrete};
\end{algorithmic}
\hspace*{0.02in} {\bf When $k\geq 1$:}
\begin{algorithmic}[1] 
\STATE All agents update according to the normal rule \dref{iter_from_1_discrete}.
\end{algorithmic}
\end{algorithm}


\begin{theorem}\label{thm1_discrete}
By Algorithm 1, accurate average consensus is achieved.
\end{theorem}

\begin{proof}
Firstly, we prove that the first iteration \eqref{iter_0_discrete} does not change the sum of agents' states.
Define the perturbation matrix as $\mathbf{P}=[p_i^{(j)}]\in\mathbb{R}^{N\times N}$,
where $p_i^{(j)}$ is defined to be 0 if $(i,j)\notin \mathcal{E}$. Note that $p_i^{(i)}=-\sum_{j\in\mathcal{N}_i^{out}}p_i^{(j)}$, we have $
\mathbf{P}\mathbf{1}_N=\mathbf{0},$ where $\mathbf{1}_N$ is a $N\times 1$ vector with unitary elements. By definition, we rewrite the first iteration \dref{iter_0_discrete} as
\begin{equation}\label{iter_0_matrix_form_discrete}
\mathbf{x}[1]=(\mathbf{I}-\epsilon_1\mathbf{L})\mathbf{x}[0]
  +\epsilon_1\mathbf{P}^T\mathbf{1}_N,
\end{equation}
where $\mathbf{x}=[x_1,\cdots,x_N]^T$, $\mathbf{I}$ denotes the identity matrix, and $\mathbf{L}=[l_{ij}]\in\mathbb{R}^{N\times N}$ represents the Laplacian matrix, defined as $l_{ii}=\sum_{j\neq i}a_{ij}$ and $l_{ij}=-a_{ij}$ if $i\neq j$.. Note that $\mathbf{L}^T\mathbf{1}_N=\mathbf{0}$ as the graph is balanced, we have
\begin{equation}\label{invariant_sum_discrete}
\begin{aligned}
\mathbf{1}_N^T \mathbf{x}[1]
  =\mathbf{1}_N^T(\mathbf{I}-\epsilon_1\mathbf{L})\mathbf{x}[0]
  +\epsilon_1\mathbf{1}_N^T\mathbf{P}^T\mathbf{1}_N
  =\mathbf{1}_N^T\mathbf{x}[0],
\end{aligned}
\end{equation}
implying that the sum of agents' states is unchanged after the first iteration.

From $k=1$, the agents implement the normal algorithm in Lemma \ref{lemma_normal_rule_discrete}. Then, it follows from Lemma \ref{lemma_normal_rule_discrete} that the accurate average consensus can be achieved.
\hfill $\blacksquare$
\end{proof}

\begin{remark}\label{remark_1_discrete}
The injected perturbations are locally zero-sum, i.e., $\sum_{j\in\{i\}\bigcup\mathcal{N}_i^{out}}p_i^{(j)}=0$. This is vital to guaranteeing the accuracy of the final convergence value.
\end{remark}

\subsection{Privacy Preservation Against Internal Honest-but-curious Agents}
The privacy preservation property against internal honest-but-curious agents is firstly analyzed. We assume that there exists only one internal honest-but-curious agent and it is represented by $m$. Nevertheless, it is worth noting that our results can be easily extended to collaborative agents by regarding them as a super node.

What information can the attacker obtain should be defined first. It is assumed that every internal agent (curious or not) can store all the information transmitted to it and knows the graph topology and the parameters $\epsilon_1$ and $\epsilon_2$.
Thus, for any agent $i$, we define the set of all information that it can access as its information set $\Phi_i$, which is defined as
\begin{equation}\label{information_set_internal_discrete}
\begin{aligned}
\Phi_i=&\{\mathcal{G};\mbox{ } \epsilon_1,\epsilon_2;\mbox{ } y_j^{(i)}[0],j\in\mathcal{N}_i^{in};
\mbox{ } x_j[k],k\geq 1, \\
&j\in\mathcal{N}_i^{in};\mbox{ } x_i[k],k\geq 0;\mbox{ } p_i^{(j)},j\in \mathcal{N}_i^{out},  \\
&\mbox{ } the \mbox{ } form \mbox{ } of \mbox{ } Algorithm \mbox{ } \ref{alg:our_alg_discrete}\}.
\end{aligned}
\end{equation}

Inspired by \cite{Kia}, we present a lemma before moving forward. Without loss of generality, we arbitrarily choose an agent in the network and label it as agent 1. As the graph is strongly connected and balanced, its in-neighbor set $\mathcal{N}_1^{in}$ is not empty. We arbitrarily choose an agent from $\mathcal{N}_1^{in}$ and label it as agent 2, and let $\mathcal{V}^3=\mathcal{V}\backslash\{1,2\}$ be the set of other agents. Denoted by $\mathbf{x}_3[k]$ the aggregated states of agents in $\mathcal{V}^3$. In a compatible manner, $\mathbf{L},\mathbf{A}$ and $\mathbf{P}$ are partitioned to the following subblock matrices:
\begin{equation}\label{L_and_P_discrete}
  \mathbf{L}=\begin{bmatrix}
               d_1 & -1 & -\mathbf{A}_{13} \\
               -a_{21} & d_2 & -\mathbf{A}_{23} \\
               -\mathbf{A}_{31} & -\mathbf{A}_{32} & \mathbf{L}_{33}
             \end{bmatrix},
  \mathbf{P}=\begin{bmatrix}
                p_1^{(1)} & p_1^{(2)}& \mathbf{P}_1^{(3)} \\
                p_2^{(1)} & p_2^{(2)} & \mathbf{P}_2^{(3)} \\
                \mathbf{P}_3^{(1)} & \mathbf{P}_3^{(2)} & \mathbf{P}_3^{(3)}
              \end{bmatrix}.
\end{equation}

\begin{lemma}\label{lemma_my_internal_discrete}
Consider two implementations of our algorithm. In the first implementation, agents' initial states are $x_1[0],x_2[0],\cdots,x_N[0]$ and the generated perturbation signals are $p_i^{(j)}$s. In the second implementation, agents' initial states are
\begin{equation}\label{second_implementation_lemma_1_ini_discrete}
\begin{aligned}
&\tilde{x}_1[0]\in\mathbb{R},\mbox{ } \tilde{x}_2[0]=x_1[0]+x_2[0]-\tilde{x}_1[0],\\
&\mathbf{\tilde{x}}_3[0]=\mathbf{x}_3[0],
\end{aligned}
\end{equation}
and the perturbations are given as
\begin{equation}\label{second_implementation_lemma_1_perturbation_discrete}
\begin{aligned}
&\tilde{p}_1^{(1)}=p_1^{(1)}+d_1(\tilde{x}_1[0]-x_1[0]),\\
&\tilde{p}_1^{(2)}=p_1^{(2)}-a_{21}(\tilde{x}_1[0]-x_1[0]),\\
&\mathbf{\tilde{P}}_1^{(3)}=\mathbf{P}_1^{(3)}-\mathbf{A}_{31}^T(\tilde{x}_1[0]-x_1[0]),\\
&\tilde{p}_2^{(1)}=p_2^{(1)}+(\frac{1}{\epsilon_1}-1)(\tilde{x}_2[0]-x_2[0]),\\
&\tilde{p}_2^{(2)}=p_2^{(2)}+(d_{2}-\frac{1}{\epsilon_1})(\tilde{x}_2[0]-x_2[0]),\\
&\mathbf{\tilde{P}}_2^{(3)}=\mathbf{P}_2^{(3)}-\mathbf{A}_{32}^T(\tilde{x}_2[0]-x_2[0]),\\
&\mathbf{\tilde{P}}_3^{(j)}=\mathbf{P}_3^{(j)},\mbox{ }j\in{1,2,3}.
\end{aligned}
\end{equation}

Then, in both implementations, accurate average consensus can be achieved and the the convergence values are the same, i.e., $\lim_{k\rightarrow\infty}x_i[k]=\tilde{x}_i[k]$. Moreover, if agent $i\in\mathcal{V}^3$, then in both implementations agent $i$'s information sets are the same, i.e., $\Phi_i=\tilde{\Phi}_i$.
\end{lemma}

\begin{figure}[htp]
\centering
\includegraphics[width=0.68\linewidth,height=0.23\linewidth]{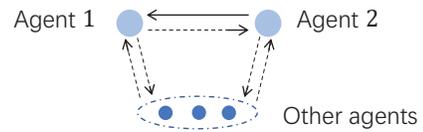}
\caption{The graph described in Lemma \ref{lemma_my_internal_discrete}. The dashed arrows means that these edges may or may not really exist, while the solid arrow means that this edge does exist.}
  \label{fig_graph_lemma_discrete}
\end{figure}

\begin{proof}
By \eqref{second_implementation_lemma_1_perturbation_discrete}, we see that $\mathbf{\tilde{P}}\mathbf{1}_N=0$ holds, implying that $\tilde{p}_i^{(j)}$s are a group of possible perturbation signals. Then by Theorem \ref{thm1_discrete}, in both implementations accurate average consensus can be guaranteed. Moreover, from \eqref{second_implementation_lemma_1_ini_discrete} we know $\mathbf{1}_N^T\mathbf{\tilde{x}}[0]=\mathbf{1}_N^T\mathbf{x}[0]$. Therefore, $\lim_{k\rightarrow\infty}x_i[k]=\tilde{x}_i[k]$ holds.

For any $i\in\mathcal{V}^3$, the key point to compare  $\Phi_i$ and $\tilde{\Phi}_i$ is to focus on the transmitted messages when $k=0$. We denote the message transmitted from agent $i$ to agent $j$ when $k=0$ in the two implementations as $y_i^{(j)}[0]$ and $\tilde{y}_i^{(j)}[0]$, respectively. Define $\delta y_i^{(j)}[0]=\tilde{y}_i^{(j)}[0]-y_i^{(j)}[0]$. By \eqref{second_implementation_lemma_1_ini_discrete} and \eqref{second_implementation_lemma_1_perturbation_discrete}, we have
\begin{equation}\label{delta_y_zero_except_(2,1) _discrete}
\delta {y}_i^{(j)}=\begin{cases}
  0, & \mbox{if } (i,j)\in\mathcal{E}\backslash(2,1), \\
  \frac{1}{\epsilon_1}(\tilde{x}_2[0]-x_2[0]), & \mbox{if } (i,j)=(2,1),
\end{cases}
\end{equation}
implying that for any $i\in\mathcal{V}^3$, all the messages it transmits and receives when $k=0$ are the same in the two implementations.

Now we focus on how the two systems evolve. Define $\delta\mathbf{x}=\mathbf{\tilde{x}}-\mathbf{x}$ and $\Delta \mathbf{P}=\mathbf{\tilde{P}}-\mathbf{P}$. By \eqref{iter_0_matrix_form_discrete}, \eqref{second_implementation_lemma_1_ini_discrete} and \eqref{second_implementation_lemma_1_perturbation_discrete}, we have
\begin{equation}\label{delta_x_lemma_1_discrete}
\begin{aligned}
&\delta \mathbf{x}[1]= \mathbf{\tilde{x}}[1]-\mathbf{x}[1]\\
=&(\mathbf{I}-\epsilon_1\mathbf{L})\delta\mathbf{x}[0]
  +\epsilon_1\Delta\mathbf{P}^T\mathbf{1}_N \\
=&\begin{bmatrix}
    1-\epsilon_1d_1 & \epsilon_1 & \epsilon_1\mathbf{A}_{13} \\
    \epsilon_1a_{21} & 1-\epsilon_1d_2 & \epsilon_1\mathbf{A}_{23} \\
    \epsilon_1\mathbf{A}_{31} & \epsilon_1\mathbf{A}_{32} & \mathbf{I}-\epsilon_1\mathbf{L}_{33}
  \end{bmatrix}
  \begin{bmatrix}
    \delta x_1[0] \\
    \delta x_2[0] \\
    \mathbf{0}
  \end{bmatrix}+ \\
+&\epsilon_1\begin{bmatrix}
    d_1\delta x_1[0] & -a_{21}\delta x_1[0] &  -\mathbf{A}_{31}^T\delta x_1[0] \\
    (\frac{1}{\epsilon_1}-1)\delta x_2[0] & (d_2-\frac{1}{\epsilon_1})\delta x_2[0] & -\mathbf{A}_{32}^T\delta x_2[0] \\
    \mathbf{0} & \mathbf{0} & \mathbf{0}
  \end{bmatrix}^T\mathbf{1}_N \\
=&\mathbf{0},
\end{aligned}
\end{equation}
implying that $\mathbf{\tilde{x}}[1]=\mathbf{x}[1]$. Hereafter, the two systems are exactly the same at every iteration, so does the transmitted messages in the two implementations from $k\geq 1$.

The proof is then completed.
\hfill $\blacksquare$
\end{proof}

\begin{remark}
As shown in Lemma \ref{lemma_my_internal_discrete}, the main difference between the two implementations is that $x_1[0]\neq\tilde{x}_1[0]$ and $x_2[0]\neq\tilde{x}_2[0]$. However, this difference does not influence those agents in $\mathcal{V}^3$ at all, as their trajectories of states and their collected information are unchanged. This fact actually implies that all agents in $\mathcal{V}^3$ cannot distinguish any variation of $x_1[0]$ and $x_2[0]$. For any agent $i\in\mathcal{V}^3$ and a given information set $\Phi_i$, there exists infinite number of possible initial conditions and corresponding perturbation signals for agent 1 and 2 which can lead to the same information set $\Phi_i$.
\end{remark}

\begin{theorem}\label{thm2_discrete}
For any agent $\upsilon\in\mathcal{V}$, the internal honest-but-curious agent $m$ can evaluate its initial value if and only
\begin{equation}\label{thm2_condition_discrete}
\begin{cases}
  &d_\upsilon =1,\\
  &(m,\upsilon),\mbox{ } (\upsilon,m)\in \mathcal{E}.
\end{cases}
\end{equation}
\end{theorem}

\begin{proof}
(Necessity) When \eqref{thm2_condition_discrete} holds, agent $\upsilon$ can only communicate with the internal honest-but-curious agent $m$ (see Fig. \ref{fig_graph_fail_condition_discrete}). From Algorithm 1, we have
\begin{equation}\label{fail_thm_1_discrete}
\begin{aligned}
&y_\upsilon^{(m)}[0]=x_\upsilon[0]+p_\upsilon^{(m)}, \\
&x_\upsilon[1]=x_\upsilon[0]+\epsilon_1(y_m^{(\upsilon)}[0]-x_\upsilon[0])+\epsilon_1p_\upsilon^{(\upsilon)}.
\end{aligned}
\end{equation}
Combing \eqref{fail_thm_1_discrete} and $p_\upsilon^{(\upsilon)}=-p_\upsilon^{(m)}$, agent $m$ can compute $x_\upsilon[0]$ with the following equation:
\begin{equation}\label{observer_discrete}
x_\upsilon[0]=x_\upsilon[1]+\epsilon_1(y_\upsilon^{(m)}[0]-y_m^{(\upsilon)}[0]).
\end{equation}

(Sufficiency) Now we are going to show that agent $\upsilon$'s privacy is preserved if \eqref{thm2_condition_discrete} does not hold. Note that agent $m$'s information set $\Phi_m$ plays an important role in the analysis. By definition, $\Phi_m$ is the set of agent $m$'s collected information and is the only thing that agent $m$ can make use of to evaluate other agents' privacy. It is natural to see that if $\Phi_m$ can be exactly the same when agent $\upsilon$'s initial state is an arbitrary value, then there is no doubt that agent $m$ cannot distinguish between these possible initial values of agent $\upsilon$, implying that agent $m$ cannot uniquely determine the value of $x_\upsilon[0]$.

Note that $N\geq 3$ and the graph is strongly connected and balanced, if \eqref{thm2_condition_discrete} does not hold, at least one of the following two cases holds:

Case 1: There exists an agent $\omega$ such that $(\omega,\upsilon)\in\mathcal{E}$. We can regard agent $\omega$ and agent $\upsilon$ as the agent 2 and the agent 1 in Lemma \ref{lemma_my_internal_discrete}. It follows from Lemma \ref{lemma_my_internal_discrete} that all other agents' information sets, including agent $m$'s information set $\Phi_m$, can stay exactly unchanged even when agent $\upsilon$'s initial state changes from $x_\upsilon[0]$ to $\tilde{x}_\upsilon[0]$, where $\tilde{x}_\upsilon[0]$ can be arbitrarily chosen from $\mathbb{R}$. Thus, there exists infinite number of possible initial values of agent $\upsilon$ that can lead to the same information set $\Phi_m$. We then conclude that agent $m$ cannot reconstruct agent $\upsilon$'s initial state. In this sense, we say that agent $\upsilon$'s privacy is preserved under Algorithm \ref{alg:our_alg_discrete}.

Case 2: There exists an agent $\omega$ such that $(\upsilon,\omega)\in\mathcal{E}$. Following similar lines, we can prove that agent $\upsilon$'s privacy is preserved.

The proof is then completed.
\hfill $\blacksquare$
\end{proof}

\begin{figure}[htp]
\centering
\includegraphics[width=0.7\linewidth,height=0.35\linewidth]{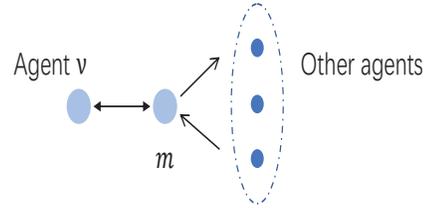}
\caption{The case where agent $\upsilon$'s privacy will be disclosed.}
  \label{fig_graph_fail_condition_discrete}
\end{figure}

\begin{remark}\label{remark_2_discrete}
Theorem \ref{thm2} states that an agent $i$'s privacy can be still preserved even when agent $i$ and all its neighbors are connected to the honest-but-curious agent, in which case the algorithms in \cite{Mo2016Privacy} and \cite{Manitara2013Privacy} will fail. Compared with \cite{Kia}, our algorithm performs better when honest-but-curious agents exist. In \cite{Kia}, an agent $i$'s privacy will be disclosed if and only if  $(i,m)\in\mathcal{E}$ and $\mathcal{N}_i^{in}\subset\mathcal{N}_m^{in}\bigcup\{m\}$. By Theorem \ref{thm2}, we can see that our algorithm is still valid even when the above condition in \cite{Kia} is violated. Therefore, our algorithm improves over those in the aforementioned related works.
\end{remark}

\begin{remark}
In this paper, the added perturbation signals are edge-based, meaning that each agent adds different perturbation signals to the information transmitted to different out-neighbors. Thus, our algorithm has more degrees of freedom on perturbation signal design than those in \cite{Mo2016Privacy,He2018privacy,Kia} where node-based method is used. It should be mentioned that
\cite{GuptaPrivacy} also proposed an edge-based algorithm, whose basic idea is to make sure that the distributions of the honest agents' initial states conditioned on their sum and the attackers' view are identical. The privacy analysis in this section is significantly different from the probability viewpoint in \cite{GuptaPrivacy}. Moreover, the privacy-preserving algorithm in the current paper is simpler than that in \cite{GuptaPrivacy} and can be applicable to a more broad class of network topologies which are not necessarily undirected.
\end{remark}

\subsection{Privacy Preservation Against External Eavesdroppers}
In this subsection, we consider another case that external eavesdroppers exist. Here we should define the information sets of these eavesdroppers. It is assumed the graph and all transmitted information between agents are accessible to them. And we assume that value of the parameter $\epsilon_1$ is not known by them. Denoted by $\Theta$ an external eavesdropper's information set, which is defined as
\begin{equation}\label{information_set_external_discrete}
\begin{aligned}
\Theta=\{&\mathcal{G} ;\mbox{ }\epsilon_2; \mbox{ } y_i^{(j)}[0],\forall (i,j)\in\mathcal{V};\mbox{ } x_i[k],\forall i\in\mathcal{V}, k\geq 1, \\
&the \mbox{ } form \mbox{ } of \mbox{ } Algorithm \mbox{ } \ref{alg:our_alg_discrete}\}.
\end{aligned}
\end{equation}

Similar with the last subsection, the following lemma is presented at first.

\begin{lemma}\label{lemma_my_external_discrete}
Assuming that there are two implementations of our algorithm. In the first one, agents' initial states are $\mathbf{x}[0]=(x_1[0],\cdots,x_N[0])^T$, the perturbation signals are $p_i^{(j)}$s, and parameters are $\epsilon_1$ and $\epsilon_2$. The initial conditions, perturbations and parameters of the second implementation are given as
\begin{equation}\label{second_implementation_lemma_2_discrete}
\begin{aligned}
\mathbf{\hat{x}}[0]&=\mathbf{x}[0]+\delta\epsilon(\mathbf{L}\mathbf{x}[0]-\mathbf{P}^T\mathbf{1}_N),\\
\hat{p}_i^{(j)}&=p_i^{(j)}-(\hat{x}_i[0]-x_i[0]),\mbox{ }\forall (i,j)\in\mathcal{E},\\
\hat{p}_i^{(i)}&=-\sum_{j\in\mathcal{N}_i^{out}}\hat{p}_i^{(j)}, \\
\hat{\epsilon}_1&=\epsilon_1+\delta\epsilon,~
\hat{\epsilon}_2=\epsilon_2,
\end{aligned}
\end{equation}
where $\delta\epsilon \in\mathbb{R}$ is a constant.

Then, in both implementations, accurate average consensus can be achieved and $\lim_{k\rightarrow\infty}x_i[k]=\hat{x}_i[k]$. Moreover, the external eavesdropper's information sets are exactly the same in the two implementations, i.e., $\Theta=\hat{\Theta}$.
\end{lemma}

\begin{proof}
We first focus on convergence. It is easy to see that $\mathbf{\hat{P}}\mathbf{1}_N=0$ holds. Thus, accurate average consensus can be achieved in both implementations according to Theorem \ref{thm1_discrete}. Moreover, note that $\mathbf{1}_N\mathbf{x}[0]=\mathbf{1}_N\hat{\mathbf{x}}[0]$, we have $\lim_{k\rightarrow\infty}x_i[k]=\hat{x}_i[k]$.

Consider the transmitted messages at the first iteration ($k=0$), it is not difficult to verify from \eqref{second_implementation_lemma_2_discrete} that $\hat{y}_i^{(j)}[0]=y_i^{(j)}[0]$, $\forall (i,j)\in\mathcal{E}$ holds. It means that all transmitted messages at the first iteration are exactly the same in the two implementations. Then we consider the case that $k\geq 1$. Define $\delta\mathbf{x}=\mathbf{\hat{x}}-\mathbf{x}$ and $\Delta \mathbf{P}=\mathbf{\hat{P}}-\mathbf{P}$. By definition and in light of \eqref{iter_0_matrix_form_discrete}, we have
\begin{equation}\label{delta_x_lemma_2_discrete}
\begin{aligned}
&\delta \mathbf{x}[1]= \mathbf{\hat{x}}[1]-\mathbf{x}[1]\\
=& (\mathbf{I}-\hat{\epsilon}_1\mathbf{L})\mathbf{\hat{x}}[0]+\hat{\epsilon}_1\mathbf{\hat{P}}^T\mathbf{1}_N-
(\mathbf{I}-\epsilon_1\mathbf{L})\mathbf{x}[0]-\epsilon_1\mathbf{P}^T\mathbf{1}_N\\
=&(\mathbf{I}-\epsilon_1\mathbf{L}-\delta \epsilon\mathbf{L})(\mathbf{x}[0]+\delta \mathbf{x}[0])+
(\epsilon_1+\delta\epsilon)\\&\times (\mathbf{P}+\Delta \mathbf{P})^T\mathbf{1}_N-(\mathbf{I}-\epsilon_1\mathbf{L})\mathbf{x}[0]-\epsilon_1\mathbf{P}^T\mathbf{1}_N \\
=&\delta \mathbf{x}[0]-\delta \epsilon(\mathbf{L}\mathbf{x}[0]-\mathbf{P}^T\mathbf{1}_N)
-\hat{\epsilon}(\mathbf{L}\delta \mathbf{x}[0]-\Delta \mathbf{P}^T\mathbf{1}_N).
\end{aligned}
\end{equation}
It follows from \eqref{second_implementation_lemma_2_discrete} that
\begin{equation}\label{zero_term_lemma_2_discrete}
\begin{aligned}
&\mathbf{L}\delta \mathbf{x}[0]-\Delta \mathbf{P}^T\mathbf{1}_N \\
=& \mathbf{L}\delta \mathbf{x}[0] \\
&-\begin{bmatrix}
  d_1\delta x_1[0] & -a_{12}\delta x_2[0] & \cdots & -a_{1N}\delta x_N[0] \\
  -a_{21}\delta x_1[0] & d_2\delta x_2[0] & \cdots & -a_{2N}\delta x_N[0] \\
  \vdots & \vdots &  & \vdots  \\
  -a_{N1}\delta x_1[0] & -a_{2N}\delta x_2[0] & \cdots & d_N\delta x_N[0]
\end{bmatrix}\mathbf{1}_N \\
=&\mathbf{0}.
\end{aligned}
\end{equation}
Substituting \eqref{second_implementation_lemma_2_discrete} and \eqref{zero_term_lemma_2_discrete} into \eqref{delta_x_lemma_2_discrete}, we have
\begin{equation}\label{delta_x_lemma_2_equal_zero_discrete}
\delta \mathbf{x}[1]=\mathbf{0},
\end{equation}
implying that $\hat{x}_i[1]=x_i[1],\mbox{ } \forall i\in\mathcal{V}$. It means that from $k=1$, the two systems are exactly the same at every iteration. Thus, for $k\geq 1$, the transmitted signals are also identical in the two implementations.

Then it is easy to see $\Theta=\hat{\Theta}$ holds.
\hfill $\blacksquare$
\end{proof}

\begin{theorem}\label{thm3_discrete}
If all perturbation signals $p_i^{(j)}$s are randomly chosen from $\mathbb{R}$, under Algorithm \ref{alg:our_alg_discrete} all agents' privacy is preserved, in spite of the existence of external eavesdroppers.
\end{theorem}

\begin{proof}
The proof of this theorem is similar to the proof of Theorem \ref{thm2_discrete}. The above Lemma \ref{lemma_my_external_discrete} shows that even when $\mathbf{x}[0]$ is changed to $\mathbf{x}[0]+\delta\epsilon(\mathbf{L}\mathbf{x}[0]-\mathbf{P}^T\mathbf{1}_N)$, the information set of the external attacker may stay unchanged. Let $\eta=[\eta_1,\cdots,\eta_N]^T=\mathbf{L}\mathbf{x}[0]-\mathbf{P}^T\mathbf{1}_N$. Then for any agent $i$, even if its initial state changes from $x_i[0]$ to $x_i[0]+\delta\epsilon\cdot \eta_i$, it is possible that $\Theta$ is not influenced at all. Note that $\delta\epsilon\in\mathbb{R}$ is an arbitrary value, then if $\eta_i\neq 0$, $x_i[0]+\delta\epsilon\cdot \eta_i$ can be any arbitrary value in $\mathbb{R}$. In this case, there exists infinite number of possible initial values of agent $i$ and the external eavesdropper cannot reconstruct the value of $x_i[0]$ so agent $i$'s privacy is preserved. Note that when all $p_i^{(j)}$s are randomly chosen from $\mathbb{R}$, the possibility of $\eta_i=0$ is zero. The proof is then completed. \hfill $\blacksquare$
\end{proof}

\begin{remark}\label{remark_4_discrete}
It should be noted that many works based on injecting correlated-noise like \cite{Mo2016Privacy} and \cite{He2018privacy} are vulnerable to the external eavesdropper attackers. By contrast, our algorithm is valid when they exist. It is also worth noting that those works that rely on cryptology, e.g., there in \cite{WangCryptography} and \cite{YuezuLv2019accurate}, can also deal with this type of attackers well, as it is hard for the attackers to decrypt the encrypted messages. However, these methods suffer from huge costs on computation and communication, while the algorithm in this paper is light-weight.
\end{remark}


\section{Continuous-time Privacy Preserving Updating Rule and Privacy Analysis}
This section considers continuous-time privacy preserving average consensus, which is the continuous-time counterpart of the discrete-time algorithm in Section 3.

Our continuous-time privacy-preserving average consensus algorithm is divided into two phases, one of which is for obfuscation and the other is for reaching consensus. Assume that all agents know a specific instant $t_0\in(0,\infty)$. When $t\in [0,t_0]$, each agent generates some perturbation signals and then adds them to the edges leading out of it. Specifically, for each agent $i$, if $(i,j)\in\mathcal{E}$, then agent $i$ generates a bounded perturbation signal $p_i^{(j)}(t):[0,t_0]\rightarrow \mathbb{R}$ and transmit $y_i^{(j)}(t)=x_i(t)+p_i^{(j)}(t)$ to its out-neighbor, agent $j$. There are some constraints on those perturbation signals, and we refer to them as admissible perturbation signals.

\begin{definition}
We say a group of perturbation signals $p_i^{(j)}(t):[0,t_0]\rightarrow \mathbb{R}$,  generated by agent $i$, is admissible, if
 \romannumeral1) If $j\notin \mathcal{N}_i^{out}\cup\{i\}$, $p_i^{(j)}(t)\equiv 0$;
 \romannumeral2) All $p_i^{(j)}(t)$s are bounded and continuous on the interval $[0,t_0]$; \romannumeral3) $p_i^{(i)}(t)=-\sum_{j\in\mathcal{N}_i^{out}}p_i^{(j)}(t)$, $\forall t\in [0,t_0]$, $\forall i\in\mathcal{V}$.
\end{definition}

When $t\in[0,t_0]$, all the agents update their states as
\begin{equation}\label{x_update_with_perturbation_continuous}
\begin{aligned}
\dot{x}_i(t)=-c_1\sum_{j=1}^{N}a_{ij}(x_i(t)-y_j^{(i)}(t))+c_1 p_i^{(i)}(t),
\end{aligned}
\end{equation}
where $c_1\in(-\infty,0)\cup(0,\infty)$ is a constant that is known to all agents. Then, when $t>t_0$, each agent updates following the normal rule in Lemma \ref{lemma_normal_rule_continuous}, i.e.,
\begin{equation}\label{x_update_after_perturbation_continuous}
\dot{x}_i(t)=-c_2\sum_{j=1}^Na_{ij}(x_i(t)-x_j(t)), \mbox{ }\forall t > t_0,
\end{equation}
where $c_2>0$ is a constant.

We summarize the proposed algorithm in Algorithm \ref{alg:our_alg_continuous}.

\begin{algorithm}[htb]
\caption{Privacy-Preserving Average Consensus Algorithm.}
\label{alg:our_alg_continuous}
\begin{algorithmic}[1] 
\STATE Each agent $i$ generates a group of admissible perturbation signals $p_i^{(j)}(t)$s,  $\forall j\in\mathcal{N}_i^{out}\bigcup\{i\}$;
\end{algorithmic}
\hspace*{0.02in} {\bf When $t\leq t_0$:}
\begin{algorithmic}[1] 
\STATE Agent $i$ transmits $y_i^{(j)}(t)=x_i(t)+p_i^{(j)}(t)$ to agent $j$, $\forall j\in\mathcal{N}_i^{out}$;
\STATE Each agent $i$ updates its state according to \dref{x_update_with_perturbation_continuous};
\end{algorithmic}
\hspace*{0.02in} {\bf When $t> t_0$:}
\begin{algorithmic}[1] 
\STATE All agents update via the normal rule \dref{x_update_after_perturbation_continuous}.
\end{algorithmic}
\end{algorithm}


\begin{theorem}\label{thm1}
Accurate average consensus can be achieved by Algorithm 1.
\end{theorem}

\begin{proof}
Firstly, we prove that the coordinated scrambling phase does not change the sum of agents' states.
Define the perturbation matrix as $\mathbf{P}(t)=[p_i^{(j)}(t)]$.
From the fact that $p_i^{(i)}(t)=-\sum_{j\in\mathcal{N}_i^{out}}p_i^{(j)}(t)$, we have
\begin{equation}\label{P_zero_sum_continuous}
\mathbf{P}(t)\mathbf{1}_N=\mathbf{0},\mbox{ }\forall t\in [0, t_0].
\end{equation}
The updating rule \dref{x_update_with_perturbation_continuous} can be written in the matrix form as
\begin{equation}\label{x_update_matrix__continuous}
\mathbf{\dot{x}}(t)=-c_1\mathbf{L}\mathbf{x}(t)+c_1\mathbf{P}^T(t)\mathbf{1}_N, \mbox{ }\forall t\in[0,t_0],
\end{equation}
where $\mathbf{x}(t)=[x_1(t),\cdots,x_N(t)]^T$. It then follows from the properties of balanced graphs and \eqref{P_zero_sum_continuous}, \eqref{x_update_matrix__continuous} that
\begin{equation}\label{sum_invariant_continuous}
\begin{aligned}
\mathbf{1}_N^T \dot{\mathbf{x}}(t)
&=-c_1\mathbf{1}_N^T\mathbf{L}\mathbf{x}(t)
+c_1\mathbf{1}_N^T\mathbf{P}^T(t)\mathbf{1}_N \\&=0, \mbox{ }\forall t\in[0,t_0],
\end{aligned}
\end{equation}
implying that the sum of all agents' states stays invariant. So we have $\mathbf{1}_N^T \mathbf{x}(t_0)=\mathbf{1}_N^T \mathbf{x}(0)$.

From $t_0$, a normal average consensus updating rule is applied, so it follows from Lemma \ref{lemma_normal_rule_continuous} that accurate average consensus can be achieved.\hfill $\blacksquare$
\end{proof}


\subsection{Privacy Preservation Against Honest-but-curious Agents}
In this subsection, we are going to evaluate how our Algorithm \ref{alg:our_alg_continuous} performs when there exists an internal honest-but-curious agent in the network. Similar with the discrete-time case, we consider the case where there exists only one honest-but-curious agent $m$. And we assume that every internal agent (curious or not) can store all the information transmitted to it and knows the graph topology, and they all have access to the parameters $c_1$ and $c_2$. Thus, an agent $i$'s information set $\mathcal{S}_i$ is
\begin{equation}\label{information_set_internal_continuous}
\begin{aligned}
\mathcal{S}_i=&\{\mathcal{G};c_1,c_2,t_0;y_j^{(i)}(t),j\in\mathcal{N}_i^{in},t\leq t_0;x_j(t),\\
&j\in\mathcal{N}_i^{in},t>t_0; x_i(t),t\geq 0;p_i^{(j)}(t),j\in \mathcal{N}_i^{out},\\
&t\leq t_0; \mbox{ } the \mbox{ } form \mbox{ } of \mbox{ } Algorithm \mbox{ } \ref{alg:our_alg_continuous}\}.
\end{aligned}
\end{equation}

Before moving forward, we partition the perturbation matrix $\mathbf{P}(t)$ to subblock matrices as
\begin{equation}\label{L_and_P}
\begin{aligned}
  \mathbf{P}(t)&=\begin{bmatrix}
                p_1^{(1)}(t) & p_1^{(2)}(t)& \mathbf{P}_1^{(3)}(t) \\
                p_2^{(1)}(t) & p_2^{(2)}(t) & \mathbf{P}_2^{(3)}(t) \\
                \mathbf{P}_3^{(1)}(t) & \mathbf{P}_3^{(2)}(t) & \mathbf{P}_3^{(3)}(t)
              \end{bmatrix},
\end{aligned}
\end{equation}
and we then present the following lemma.

\begin{lemma}\label{lemma_internal_continuous}
Consider two implementations of Algorithm 1. In the first implementation, agents' initial states are $x_1(0),x_2(0),\mathbf{x}_3(0)$ and the generated perturbation signals are $p_i^{(j)}(t)$s. In the second implementation, agents' initial states and system parameters are
\begin{equation}\label{ini_second_implementation_internal_continuous}
\begin{aligned}
&\tilde{x}_1(0)\in\mathbb{R},\mbox{ } \tilde{x}_2(0)=x_1(0)+x_2(0)-\tilde{x}_1(0),\\
&\mathbf{\tilde{x}}_3(0)=\mathbf{x}_3(0),\\
&\tilde{c}_1=c_1,\tilde{c}_2=c_2,
\end{aligned}
\end{equation}
and perturbation signals are given as
\begin{equation}\label{ps_second_implementation_internal_continuous}
\begin{aligned}
\tilde{p}_1^{(2)}(t)=&p_1^{(2)}(t)-a_{21}s(t),  \\
\mathbf{\tilde{P}}_1^{(3)}(t)=&\mathbf{P}_1^{(3)}(t)
-\mathbf{A}_{31}^Ts(t),\\
\tilde{p}_1^{(1)}(t)=&p_1^{(1)}(t)+d_1s(t),\\
\tilde{p}_2^{(1)}(t)=&p_2^{(1)}(t)-\frac{e^{-c_1t_0}(\tilde{x_1}(0)-x_1(0))}{1-e^{-c_1t_0}}, \\
\mathbf{\tilde{P}}_2^{(3)}(t)=&\mathbf{P}_2^{(3)}(t)
+\mathbf{A}_{32}^Ts(t),\\
\tilde{p}_2^{(2)}(t)=&p_2^{(2)}(t)-(d_2-1)s(t) \\
&+\frac{e^{-c_1t_0}(\tilde{x_1}(0)-x_1(0))}{1-e^{-c_1t_0}}, \\
\mathbf{\tilde{P}}_3^{(j)}(t)=&\mathbf{P}_3^{(j)}(t), j=1,2,3,
\end{aligned}
\end{equation}
where $s(t)$ is given as
\begin{equation}\label{s_fuc}
s(t)=\frac{e^{-c_1t}-e^{-c_1t_0}}{1-e^{-c_1t_0}}\cdot(\tilde{x_1}(0)-x_1(0)).
\end{equation}

Then, accurate average consensus can be achieved in both implementations and $\lim_{t\rightarrow\infty}x_i(t)=\tilde{x}_i(t)$. Moreover, if agent $i\in\mathcal{V}^3$, then in both implementations agent $i$'s information sets are the same, i.e., $\mathcal{S}_i=\tilde{\mathcal{S}}_i$.
\end{lemma}

\begin{proof}
Note that $s(t)$ is  bounded and for any $i\in\mathcal{V}$, $\tilde{p}_i^{(i)}(t)=\sum_{j\in\mathcal{N}_i^{out}}\tilde{p}_i^{(j)}(t)$ holds. Using Theorem \ref{thm1}, accurate average consensus can be achieved in both implementations. Moreover, we can easily get from \eqref{ini_second_implementation_internal_continuous} that $\mathbf{1}_N^T\tilde{\mathbf{x}}(0)=\mathbf{1}_N^T\mathbf{x}(0)$, meaning that  $\lim_{t\rightarrow\infty}x_i(t)=\tilde{x}_i(t)$.

Define two variables $\delta x_i(t)=\tilde{x}(t)-x_i(t)$, $\delta p_i^{(j)}(t)=\tilde{p}_i^{(j)}(t)-p_i^{(j)}(t)$ and let $\Delta \mathbf{x}(t)=[\delta x_1(t),\cdots,\delta x_N(t)]^T$, $\Delta\mathbf{P}(t)=[\delta p_i^{(j)}(t)]$. By \eqref{x_update_matrix__continuous}, we have
\begin{equation}\label{delta_x_eq_continuous}
\Delta\mathbf{\dot{x}}(t)=-c_1\mathbf{L}\Delta\mathbf{ x}(t)
+c_1\Delta\mathbf{ P}^T(t)\mathbf{1}_N,\mbox{ } \forall t\in[0,t_0].
\end{equation}
Substituting \eqref{ini_second_implementation_internal_continuous} and \eqref{ps_second_implementation_internal_continuous} into \eqref{delta_x_eq_continuous}, we obtain that $\Delta\mathbf{ x}(t)$ satisfies the following equation:
\begin{equation}\label{delta_x_dy_continuous}
\begin{aligned}
\begin{bmatrix}
  \delta \dot{x}_1(t) \\
  \delta \dot{x}_2(t) \\
  \Delta\mathbf{\dot{x}}_3(t)
\end{bmatrix}
=c_1
\begin{bmatrix}
-d_1 & 1 & \mathbf{A}_{13} \\
 a_{21} & - d_2 & \mathbf{A}_{23} \\
\mathbf{A}_{31} & \mathbf{A}_{32} & -\mathbf{L}_{33}
\end{bmatrix}
\begin{bmatrix}
  \delta x_1(t) \\
  \delta x_2(t) \\
  \Delta\mathbf{ x}_3(t)
\end{bmatrix}  \\
+c_1\begin{bmatrix}
   d_1s(t)-\frac{e^{-c_1t_0}\delta{x}_1(0)}{1-e^{-c_1t_0}}\\
   -(a_{21}+d_2-1)s(t)+\frac{e^{-c_1t_0}\delta{x}_1(0)}{1-e^{-c_1t_0}}\\
   -\mathbf{A}_{31}s(t)+\mathbf{A}_{32}s(t)
 \end{bmatrix}
\end{aligned}
\end{equation}
with the following initial condition:
\begin{equation}\label{ini_delta_x_continuous}
\begin{aligned}
\delta x_1(0)\in\mathbb{R},\delta x_2(0)=-\delta x_1(0),
\Delta\mathbf{ x}_3(0)=\mathbf{0}.
\end{aligned}
\end{equation}
It is not difficult to verify that
\begin{equation}\label{delta_x_exp_continuous}
\begin{aligned}
\delta x_1(t)&=s(t),\delta x_2(t)=-s(t),\\
\Delta\mathbf{ x}_3(t)&=0,\quad \forall t\in[0,t_0].
\end{aligned}
\end{equation}
is  the solution to \eqref{delta_x_dy_continuous} with the initial condition \eqref{ini_delta_x_continuous}.

Now we are ready to compare $\mathcal{S}_i$ to $\tilde{\mathcal{S}}_i$ ($i\in\mathcal{V}^3$). The key point is to analyze the difference between $\tilde{y}_i^{(j)}(t)$ and $y_i^{(j)}(t)$. Note that when $t\in[0,t_0]$, the transmitted message from agent $i$ to agent $j$ is $y_i^{(j)}(t)=x_i(t)+p_i^{(j)}(t)$, $\forall (i,j)\in\mathcal{E}$. It follows from \eqref{ps_second_implementation_internal_continuous} and \eqref{delta_x_exp_continuous} that $\forall t\in[0,t_0]$,
\begin{equation}\label{y_diff_continuous}
\begin{aligned}
\tilde{y}_i^{(j)}(t)=\begin{cases}
y_i^{(j)}(t), & \mbox{if } (i,j)\in\mathcal{E}\backslash(2,1),\\
y_2^{(1)}(t)-\frac{e^{-c_1t}\delta x_1(0)}{1-e^{-c_1t_0}}, & \mbox{if } (i,j)=(2,1).
\end{cases}
\end{aligned}
\end{equation}
Equation \eqref{y_diff_continuous} shows that when $t\in[0,t_0]$, only the messages transmitted from agent 2 to agent 1 are different in the two implementations. From $t=t_0$, no perturbation signals are injected, so all agents would send their real states to their out-neighbors. Note that from \eqref{delta_x_exp_continuous} we have
\begin{equation}\label{t0_same_internal_continuous}
\delta x_1(t_0)=s(t_0)=0, \mbox{ }\delta x_2(t_0)=-s(t_0)=0,
\end{equation}
implying that $\mathbf{\tilde{x}}(t_0)=\mathbf{x}(t_0)$. Then it is obvious that from $t=t_0$, $\mathbf{\tilde{x}}(t)=\mathbf{x}(t)$ always holds for all $t\geq t_0$. This means that all the messages transmitted in the two scenarios are exactly the same since $t=t_0$.

To sum up, when $t\in[0,t_0]$, all messages in the network except the one transmitted from agent 2 to agent 1 are identical in the two implementations. When $t>t_0$, there is no difference between those transmitted messages in the two scenarios. Thus, $\mathcal{S}_i=\tilde{\mathcal{S}}_i, \forall i\in\mathcal{V}^3$.\hfill $\blacksquare$
\end{proof}

Now we are ready to analyze Algorithm \ref{alg:our_alg_continuous}'s privacy property against internal attackers.

\begin{theorem}\label{thm2}
Let $\upsilon$ be an agent in the network, then the internal honest-but-curious agent $m$ can reconstruct its initial value if and only if the following condition holds:
\begin{equation}\label{thm2_condition}
\begin{aligned}
&d_\upsilon =1,  \\
&(m,\upsilon), (\upsilon,m)\in \mathcal{E}.
\end{aligned}
\end{equation}
\end{theorem}

\begin{proof}
(Necessity) If \eqref{thm2_condition} holds, then the honest-but-curious agent $m$ is the only agent that communicates with agent $\upsilon$ (see Fig. \ref{fig_graph_fail_condition_discrete}). For $t\leq t_0$, we have
\begin{equation}\label{fail_dy_continuous}
\dot{x}_\upsilon(t)=c_1(y_m^{(\upsilon)}(t)-x_\upsilon(t))+c_1p_\upsilon^{(\upsilon)}(t),
\end{equation}
and
\begin{equation}\label{fail_trans_continuous}
y_\upsilon^{(m)}(t)=x_\upsilon(t)+p_\upsilon^{(m)}(t)=x_\upsilon(t)-p_\upsilon^{(\upsilon)}(t).
\end{equation}
When $t>t_0$ no perturbations are added, so agent $m$ knows $x_\upsilon(t_0)$. Combining \eqref{fail_dy_continuous} and \eqref{fail_trans_continuous}, the honest-but-curious agent $m$ can compute $x_\upsilon(0)$ by
\begin{equation}\label{fail_val_continuous}
x_\upsilon(0)=x_\upsilon(t_0)-c_1\int_{0}^{t_0}[y_m^{(\upsilon)}(\tau)-y_\upsilon^{(m)}(\tau)]d\tau .
\end{equation}

(Sufficiency) We are going to prove that if \eqref{thm2_condition} does not hold, then the attacker $m$ cannot reconstruct agent $\upsilon$'s initial value. The proof is similar with the proof of Theorem \ref{thm2_discrete}. As $\mathcal{S}_m$ is the set of all the information that agent $m$ can access and is the only thing that it can make use of to evaluate other agents' initial states, it plays an important role in privacy analysis. If $\mathcal{S}_m$ can be exactly the same when agent $\upsilon$'s initial state is changed, it is impossible for agent $m$ to uniquely determine the value of $x_\upsilon(0)$.

Due to the fact that $N\geq 3$ and in light of Assumption 1, when \eqref{thm2_condition} does not hold, at least one of the following two cases is true.

Case 1: There exists another agent $\omega$ such that $\omega\in \mathcal{N}_\upsilon^{in}$, i.e., $(\omega,\upsilon)\in\mathcal{E}$. In this case, let agent $\omega$ be the agent 2 and agent $\upsilon$ be the agent 1 in Lemma \ref{lemma_internal_continuous}. From Lemma \ref{lemma_internal_continuous}, we know that all other agents' information sets, including agent $m$'s information set $\mathcal{S}_m$, can stay exactly unchanged even when agent $\upsilon$'s initial state changes from $x_\upsilon(0)$ to $\tilde{x}_\upsilon(0)$, where $\tilde{x}_\upsilon(0)$ can be arbitrarily chosen in $\mathbb{R}$. Thus, agent $m$ cannot reconstruct agent $\upsilon$'s initial state, implying that agent $\upsilon$'s privacy is preserved.

Case 2: There exists another agent $\omega$ such that $\omega\in \mathcal{N}_\upsilon^{out}$, i.e. , $(\upsilon,\omega)\in\mathcal{E}$. It can be similarly shown that Algorithm \ref{alg:our_alg_continuous} preserves agent $\upsilon$'s privacy.

The proof is then completed.\hfill $\blacksquare$
\end{proof}

\begin{remark}
\cite{Kia} also considered the privacy preserving average consensus problem in continuous-time case. As we have already pointed out, our algorithm performs better on privacy preservation against honest-but-curious agents. It is also worth mentioning that another feature of Algorithm \ref{alg:our_alg_continuous} is that the perturbations are only added in a finite time interval $[0,t_0]$, while they are added at every moment in \cite{Kia}.
\end{remark}

\subsection{Privacy Preservation Against External Eavesdroppers}

Algorithm \ref{alg:our_alg_continuous}'s privacy preserving performance against external eavesdroppers is another issue worthy of concern. The information sets of these eavesdroppers need to be first defined. Assume the graph and all transmitted information among neighboring agents are accessible but the parameter $c_1$ is invisible to them.
For an external eavesdropper, its information set is given by
\begin{equation}\label{information_set_external_continuous}
\begin{aligned}
\mathcal{M}=\{&\mathcal{G};t_0,c_2;y_i^{(j)}(t),\forall (i,j)\in\mathcal{E}, \forall t\in[0,t_0];\\
&x_i(t),\forall i\in\mathcal{V},\forall t>t_0; \\
&the \mbox{ } form \mbox{ } of \mbox{ } Algorithm \mbox{ } \ref{alg:our_alg_continuous}\}.
\end{aligned}
\end{equation}

Before giving the main result, we first present a lemma.

\begin{lemma}\label{lemma_external_continuous}
Consider two implementations of Algorithm \ref{alg:our_alg_continuous}. The first implementation's initial conditions are $\mathbf{x}(0)=(x_1(0),\cdots,x_N(0))^T$, the perturbation signals are $p_i^{(j)}(t)$s, and parameters are $c_1$ and $c_2$. In the second implementation, the initial condition is
\begin{equation}\label{ini_second_implementation_external_continuous}
\begin{aligned}
\mathbf{\hat{x}}(0)&=\mathbf{x}(0)+\delta c\cdot\mathbf{\xi},
\end{aligned}
\end{equation}
and perturbation signals and parameters are
\begin{equation}\label{ps_second_implementation_external_continuous}
\begin{aligned}
\hat{p}_i^{(j)}(t)&=p_i^{(j)}(t)-h_i(t),\mbox{ }\forall (i,j)\in\mathcal{E},\\
\hat{p}_i^{(i)}(t)&=-\sum_{j\in\mathcal{N}_i^{out}}\hat{p}_i^{(j)}(t)=p_i^{(i)}(t)+d_ih_i(t), \\
\hat{c}_1&=c_1+\delta c,
\hat{c}_2=c_2,
\end{aligned}
\end{equation}
where $\delta c \in\mathbb{R}$ is a constant, $\mathbf{\xi}\in\mathbb{R}^N$ is a vector and $h_i(t)$ denotes the i-th item of $\mathbf{h}(t)$, which are given as
\begin{equation}\label{xi_exp_continuous}
\begin{aligned}
\mathbf{\xi}=&-\frac{1}{c_1}[e^{-c_1\mathbf{L}t_0}\mathbf{x}(0)\\
&+c_1\int_0^{t_0}e^{-c_1\mathbf{L}(t_0-\tau)}\mathbf{P}^T(\tau)\mathbf{1}_Nd\tau-\mathbf{x}(0)],
\end{aligned}
\end{equation}
\begin{equation}\label{ht_exp_continuous}
\begin{aligned}
\mathbf{h}(t)=&\delta c \cdot\mathbf{\xi}+\frac{\delta c}{c_1}[e^{-c_1\mathbf{L}t}\mathbf{x}(0) \\
&+c_1\int_0^{t}e^{-c_1\mathbf{L}(t-\tau)}\mathbf{P}^T(\tau)\mathbf{1}_Nd\tau-\mathbf{x}(0)].
\end{aligned}
\end{equation}
Then, in both implementations, accurate average consensus can be achieved and $\lim_{t\rightarrow\infty}x_i(t)=\hat{x}_i(t)$. Moreover, the external eavesdropper's information sets are exactly the same, i.e., $\mathcal{M}=\hat{\mathcal{M}}$.
\end{lemma}

\begin{proof}
From \eqref{ps_second_implementation_external_continuous}, it is obvious that for each $i\in\mathcal{V}$, $\hat{p}_i^{(j)}(t)$s are also admissible perturbation signals. By Theorem \ref{thm1}, average consensus can be achieved in both implementations. Moreover, it can be verified by \eqref{xi_exp_continuous} that $\mathbf{\xi}=-\frac{1}{c_1}(\mathbf{x}(t_0)-\mathbf{x}(0))$. Thus, we have
\begin{equation}\label{ini_same_sum_external_continuous}
\begin{aligned}
\mathbf{1}_N^T\mathbf{\hat{x}}(0)&=\mathbf{1}_N^T\mathbf{x}(0)+\delta c\cdot\mathbf{1}_N^T\mathbf{\xi}\\
&=\mathbf{1}_N^T\mathbf{x}(0)-\frac{\delta c}{c_1}\mathbf{1}_N^T[\mathbf{x}(t_0)-\mathbf{x}(0)] \\
&=\mathbf{1}_N^T\mathbf{x}(0),
\end{aligned}
\end{equation}
where we have used the fact that $\mathbf{1}_N^T\mathbf{x}(t_0)=\mathbf{1}_N^T\mathbf{x}(0)$. \eqref{ini_same_sum_external_continuous} implies that $\lim_{t\rightarrow\infty}x_i(t)=\hat{x}_i(t)$.

Now we are going to prove that $\mathcal{M}=\hat{\mathcal{M}}$. The proof consists of two steps.

Step 1: we are going to show that when $t\in[0,t_0]$, all the transmitted information are identical in the two implementations. Let $\Delta \mathbf{x}(t)=\mathbf{\hat{x}}(t)-\mathbf{x}(t)$ and $\Delta \mathbf{P}(t)=\mathbf{\hat{P}}(t)-\mathbf{P}(t)$. Then the dynamic equation of $\Delta \mathbf{x}(t)$ is
\begin{equation}\label{delta_x_dy_external_continuous}
\begin{aligned}
\Delta \mathbf{\dot{x}}(t)=&-\hat{c}_1\mathbf{L}\mathbf{\hat{x}}(t)+\hat{c_1}\mathbf{\hat{P}}^T(t)\mathbf{1}_N \\
&+c_1\mathbf{L}\mathbf{x}(t)-c_1\mathbf{P}^T(t)\mathbf{1}_N \\
=&\hat{c}_1(-\mathbf{L}\Delta \mathbf{x}(t)+\Delta\mathbf{P}^T(t)\mathbf{1}_N)\\
&+\delta c(-\mathbf{L}\mathbf{x}(t)+\mathbf{P}^T(t)\mathbf{1}_N),
\end{aligned}
\end{equation}
with the initial condition
\begin{equation}\label{ini_delta_x_external_continuous}
\Delta \mathbf{x}(0)=\delta c\cdot \mathbf{\xi}.
\end{equation}

Note that from \eqref{ht_exp_continuous} we have $h(t)=\frac{\delta c}{c_1}[\mathbf{x}(t)-\mathbf{x}(t_0)]$. By solving the linear ordinary differential equation \eqref{delta_x_dy_external_continuous} with the initial condition \eqref{ini_delta_x_external_continuous}, we can obtain that
\begin{equation}\label{delta_x_exp_external_continuous}
\Delta \mathbf{x}(t)=\mathbf{h}(t).
\end{equation}

Combining \eqref{ps_second_implementation_external_continuous} and \eqref{delta_x_exp_external_continuous}, we obtain that for every $(i,j)\in\mathcal{E}$,
\begin{equation}\label{same_trans_1_external_continuous}
\begin{aligned}
\hat{y}_i^{(j)}(t)=&\hat{x}_i(t)+\hat{p}_i^{(j)}(t) \\
=&x_i(t)+\delta x_i(t)+p_i^{(j)}(t)-h_i(t)\\=&y_i^{(j)}(t),\mbox{ } \forall t\in[0,t_0],
\end{aligned}
\end{equation}
which means that when $t\in[0,t_0]$, all the transmitted information are identical in the two implementations.

Step 2: we are going to analyze the case where $t> t_0$. From \eqref{ht_exp_continuous} and \eqref{delta_x_exp_external_continuous} we have $\Delta \mathbf{x}(t_0)=\mathbf{0}$, i.e., $\mathbf{\hat{x}}(t_0)=\mathbf{x}(t_0)$. Then by $\hat{c}_2=c_2$, from $t=t_0$ the two systems are exactly the same at every moment so that there is no difference between the transmitted information in these two implementations.\hfill $\blacksquare$
\end{proof}

\begin{theorem}\label{thm3}
Let $\mathcal{F}$ be the set of all possible bounded and continuous functions defined on $[0,t_0]$, i.e., $\mathcal{F}=\{p(t):[0,t_0]\rightarrow \mathbb{R}|p(t)\mbox{ } is \mbox{ } bounded \mbox{ } and \mbox{ } continuous\}$. If all agents randomly choose the perturbation signals $p_i^{(j)}(t)$s $(j\in\mathcal{N}_i^{out})$ from $\mathcal{F}$, Algorithm \ref{alg:our_alg_continuous} can preserve all agents' privacy against external eavesdroppers, in the sense that the external eavesdroppers cannot construct any agent's initial value.
\end{theorem}

\begin{proof}
This theorem can be proved by following similar lines in the proof of Theorem \ref{thm2}. According to Lemma \ref{lemma_external_continuous}, the external attacker's information set $\mathcal{M}$ remains unchanged, even when $\mathbf{x}(0)$ is changed to $\mathbf{x}(0)+\delta c\cdot \mathbf{\xi}$. Let $\mathbf{\xi}=[\xi_1,\cdots,\xi_N]^T$. Then, for any agent $i$, the external eavesdroppers cannot distinguish $x_i(0)$ from $x_i(0)+ \delta c\cdot \mathbf{\xi}$. If $\xi_i\neq 0$, since $\delta c$ can be an arbitrary value, it follows that $x_i(0)+\delta c\cdot \xi_i$ can be also an arbitrary value, implying that there is infinite number of possible initial values of agent $i$ and the external eavesdropper cannot reconstruct the value of $x_i(0)$. As for the case of $\xi_i=0$, when all perturbation signals are randomly chosen from $\mathcal{F}$, the possibility of $\xi_i=0$ is zero. This completes the proof.\hfill $\blacksquare$
\end{proof}

\begin{remark}
Although external eavesdroppers are also considered in \cite{KiaArxiv} and \cite{WangStateDecomposition}, the ideas in these two papers and the current paper are different. In \cite{KiaArxiv}, preventing an eavesdropper from obtaining agents' privacy is achieved by letting the generated perturbation signals satisfy some certain global condition, which is invisible to the eavesdropper. In \cite{WangStateDecomposition}, it is achieved by assuming that the coupling weights between agents are not accessible to the eavesdropper. By contrast, based on the assumption that the system parameter $c_1$ is not known to the eavesdropper, we accomplishes this goal by making use of the extra degree of freedom on $c_1$'s choosing.
\end{remark}

\section{Numerical Simulation}

In this section, we focus on discrete-time case and use a numerical simulation to demonstrate Algorithm \ref{alg:our_alg_discrete}'s effectiveness. Consider a network shown in  Fig. \ref{fig_graph_simu_discrete}.

\begin{figure}[htp]
\centering
\includegraphics[width=0.50\linewidth,height=0.25\linewidth]{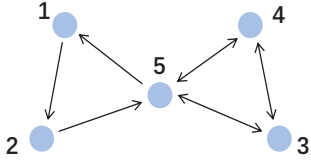}
\caption{The communication topology.}
  \label{fig_graph_simu_discrete}
\end{figure}

First, we show that Algorithm \ref{alg:our_alg_discrete} can deal with internal honest-but-curious agents. Assume that agent 5 is honest-but-curious. If the algorithm in \cite{Kia} is applied, agents 3's and 4's privacy will be revealed to agent 5. By contrast, our algorithm can preserve their privacy according to Theorem \ref{thm2_discrete}. To show this, we consider two implementations of our algorithm. In the first one, all $p_i^{(j)}$s are randomly generated, and the initial conditions and parameters are given as
\begin{equation}\label{simu_case_1_discrete}
\begin{aligned}
&x_1[0]=7,x_2[0]=3,x_3[0]=1,x_4[0]=-2,x_5[0]=-15,\\
&\epsilon_1=1, \epsilon_2=\frac{1}{3}-0.01.
\end{aligned}
\end{equation}
Let agent 4 (respectively, agent 3) be the agent 2 (respectively, agent 1) in the Lemma \ref{lemma_my_internal_discrete}. Thus, in the second implementation, the initial conditions, parameters and perturbation signals are given as
\begin{equation}\label{simu_case_2_discrete}
\begin{aligned}
&\tilde{x}_1[0]=x_1[0],\tilde{x}_2[0]=x_2[0],\tilde{x}_5[0]=x_5[0],\\
&\tilde{x}_3[0]=x_3[0]+9=10,\tilde{x}_4[0]=x_4[0]-9=-11,\\
&\tilde{\epsilon}_1=\epsilon_1, \tilde{\epsilon}_2=\epsilon_2, \\
&\tilde{p}_3^{(3)}=p_3^{(3)}+18,\tilde{p}_3^{(4)}=p_3^{(4)}-9,\tilde{p}_3^{(5)}=p_3^{(5)}-9,\\
&\tilde{p}_4^{(3)}=p_4^{(3)},\tilde{p}_4^{(4)}=p_4^{(4)}-9,\tilde{p}_4^{(5)}=p_4^{(5)}+9,\\
&\tilde{p}_i^{(j)}=p_i^{(j)},\mbox{  }otherwise.
\end{aligned}
\end{equation}

Fig. \ref{fig_internal_states_discrete} shows the state trajectories in these two implementations, from which we can observe that accurate average consensus are achieved in both cases. Moreover, it is shown that $\tilde{x}_i[k]=x_i[k]$, $\forall k\geq 1$ holds. Besides, as depicted in Fig. \ref{fig_internal_transmit_discrete}, $\tilde{y}_i^{(j)}[0]=y_i^{(j)}[0]$ always holds if $(i,j)\neq(4,3)$, implying that
even though agent 3's and 4's initial states are changed, all transmitted messages at the first iteration are identical in the two implementations, except the one transmitted from agent 4 to agent 3. Combining these two figures, we see that $\tilde{\Phi}_5=\Phi_5$, which implies that agent 5 cannot obtain agent 3's and 4's initial states.
\begin{figure}[htbp]
\centering
\subfigure{
\begin{minipage}{4.00cm}
\centering
\includegraphics[width=4.5cm]{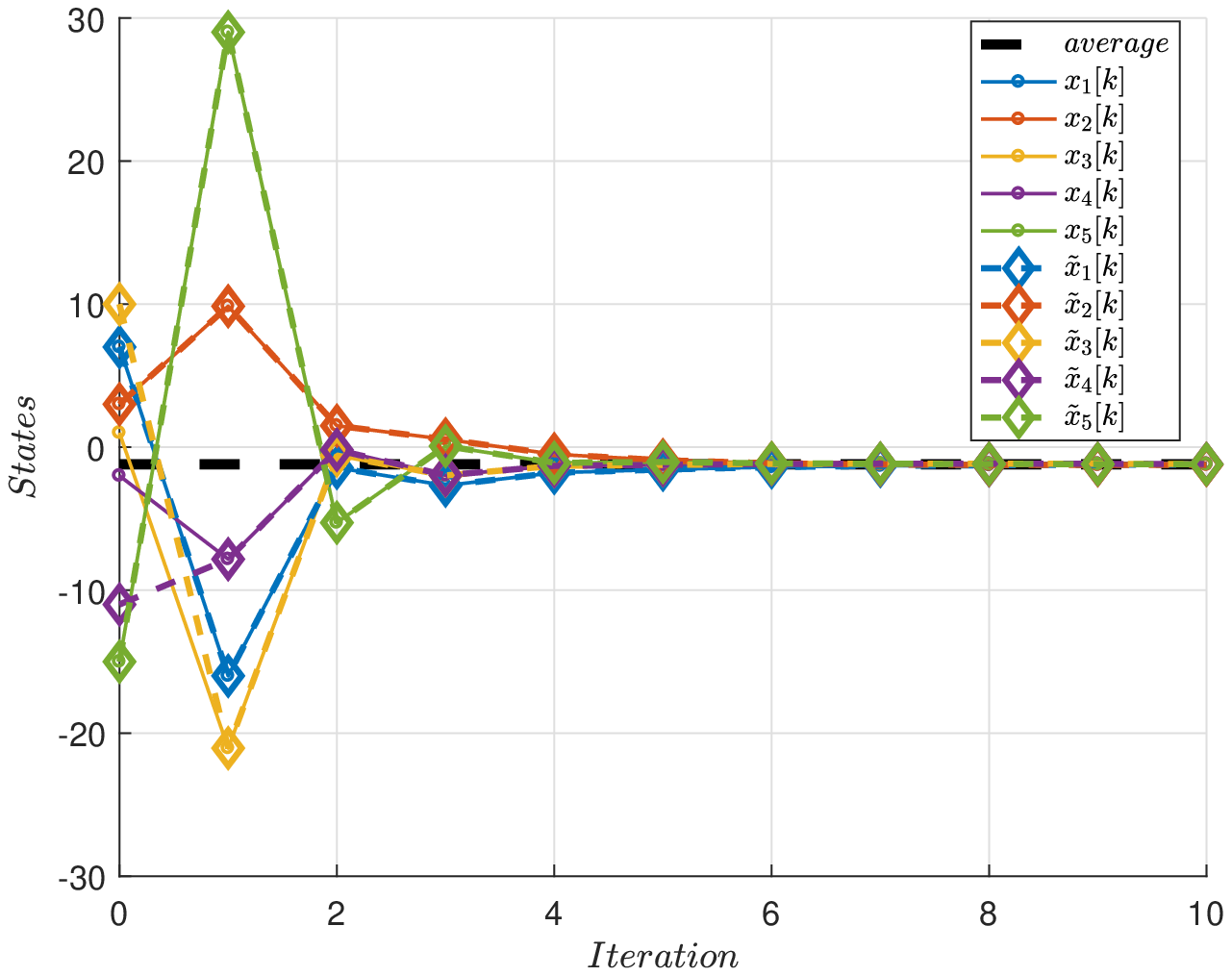}
\caption{Trajectories of the agents' states in the implementation \eqref{simu_case_1_discrete} and the implementation \eqref{simu_case_2_discrete}.}
\label{fig_internal_states_discrete}
\end{minipage}
}
\subfigure{
\begin{minipage}{4.00cm}
\centering
\includegraphics[width=4.5cm]{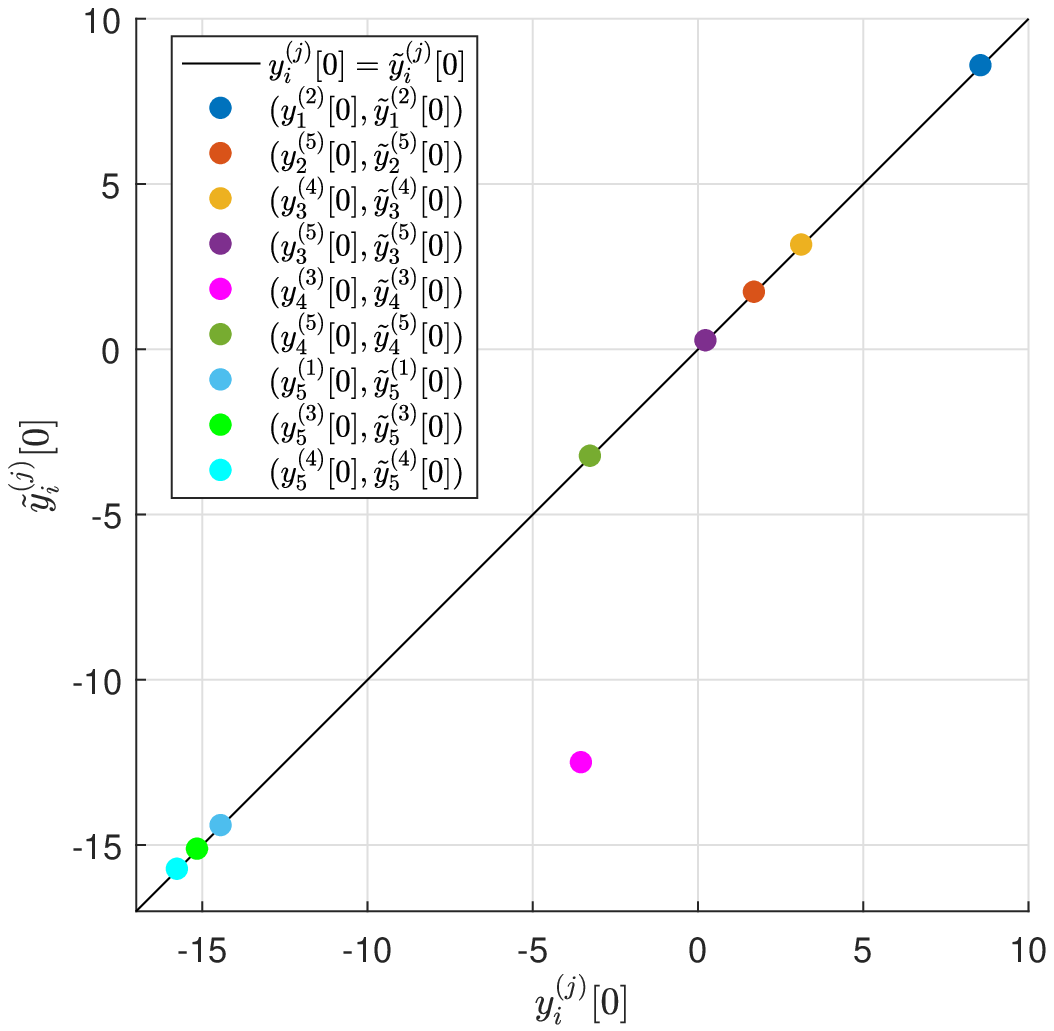}
\caption{Transmitted information at $k=0$ in the implementations \eqref{simu_case_1_discrete} and \eqref{simu_case_2_discrete}: $y_i^{(j)}[0]$s vs $\tilde{y}_i^{(j)}[0]$s.}
\label{fig_internal_transmit_discrete}
\end{minipage}
}
\end{figure}

Now we are going the show how our algorithm deals with external eavesdroppers.  According to Theorem \ref{thm3_discrete}, the attacker cannot obtain any agent's privacy. To show this, another implementation of our algorithm is considered, which is given as follows (here let $\delta \epsilon=0.8$):
\begin{equation}\label{simu_case_3_discrete}
\begin{aligned}
\mathbf{\hat{x}}[0]&=\mathbf{x}[0]+\delta\epsilon(\mathbf{L}\mathbf{x}[0]-\mathbf{P}^T\mathbf{1}_N)\\
&=\begin{bmatrix}
    25.39 & -2.48 & 18.64 & 2.67 & -50.22
  \end{bmatrix}^T, \\
\hat{\epsilon}_1&=(\epsilon_1+\delta\epsilon)=1.8,\mbox{ } \hat{\epsilon}_2=\epsilon_2, \\
\hat{p}_i^{(j)}&=p_i^{(j)}-(\hat{x}_i[0]-x_i[0]),\mbox{ }\forall (i,j)\in\mathcal{E},\\
\hat{p}_i^{(i)}&=-\sum_{j\in\mathcal{N}_i^{out}}\hat{p}_i^{(j)}.
\end{aligned}
\end{equation}

As depicted in Fig. \ref{fig_external_states_discrete}, in the above implementation \eqref{simu_case_3_discrete}, accurate average consensus is also achieved and $\hat{x}_i[k]=x_i[k]$, $\forall k\geq 1$ holds. Moreover, from Fig. \ref{fig_external_transmit_discrete} we can see that $\hat{y}_i^{(j)}[0]=y_i^{(j)}[0]$ always holds, implying that all transmitted messages when $k=0$ are identical in the implementations \eqref{simu_case_1_discrete} and \eqref{simu_case_3_discrete}. Combining these two figures, it is easy to see that the external attacker's information sets $\Theta$ and $\hat{\Theta}$ are exactly the same, even though all agents' initial state values are changed. Thus, all agents' privacy is preserved.
\begin{figure}[htbp]
\centering
\subfigure{
\begin{minipage}{4.00cm}
\centering
\includegraphics[width=4.5cm]{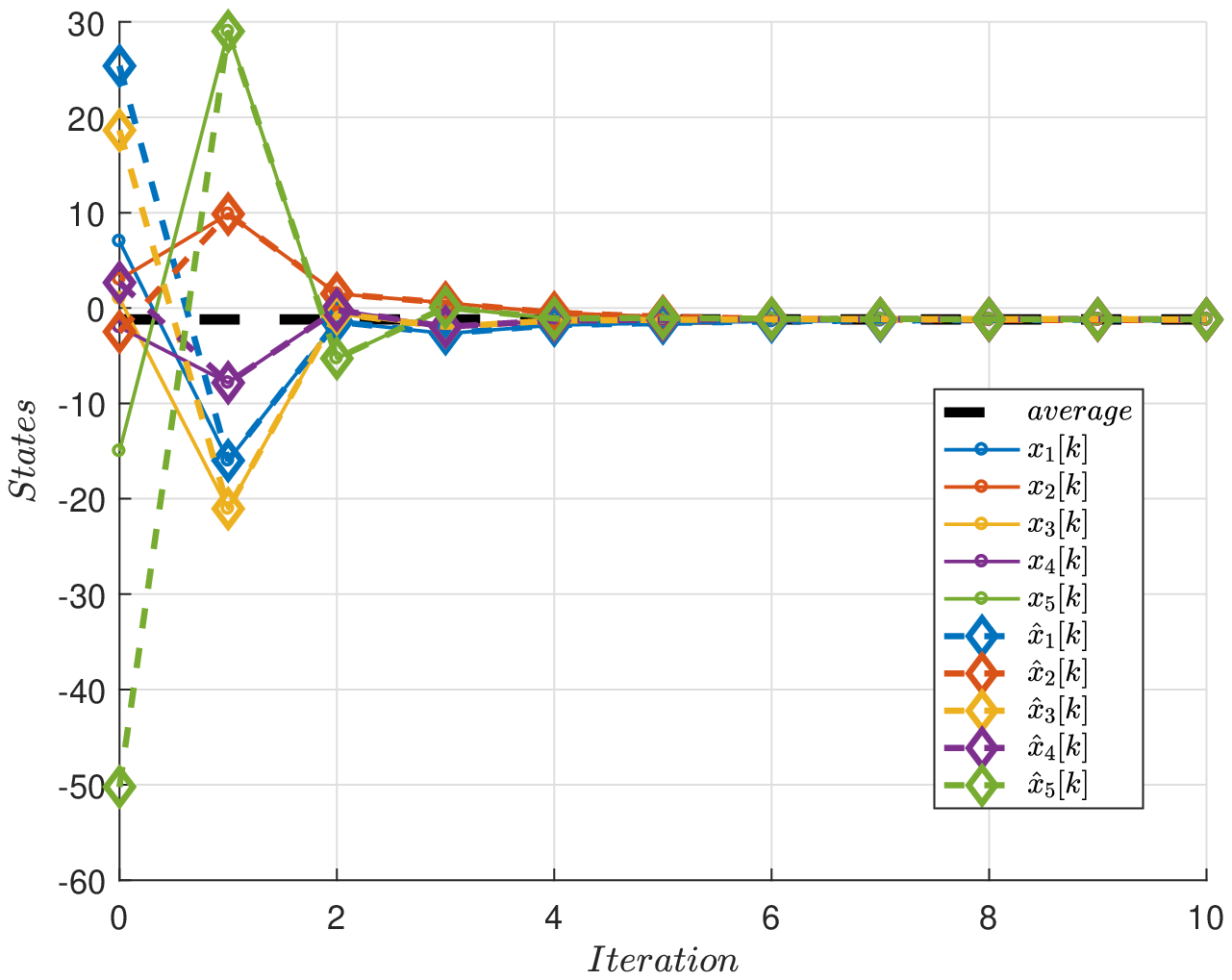}
\caption{Trajectories of the agents' states in the implementation \eqref{simu_case_1_discrete} and the implementation \eqref{simu_case_3_discrete}.}
\label{fig_external_states_discrete}
\end{minipage}
}
\subfigure{
\begin{minipage}{4.00cm}
\centering
\includegraphics[width=4.5cm]{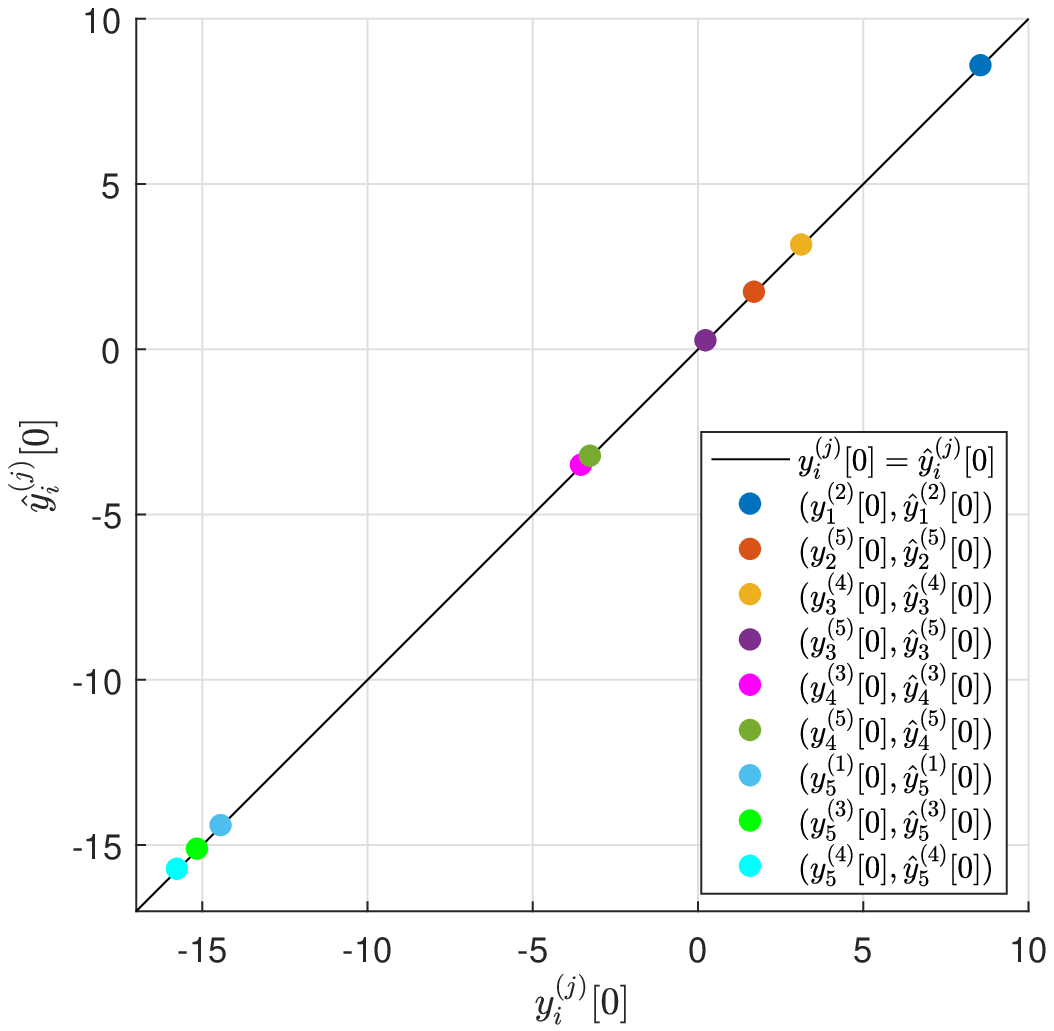}
\caption{Transmitted information at $k=0$ in the implementations \eqref{simu_case_1_discrete} and \eqref{simu_case_3_discrete}: $y_i^{(j)}[0]$s vs $\hat{y}_i^{(j)}[0]$s.}
\label{fig_external_transmit_discrete}
\end{minipage}
}
\end{figure}



\section{Conclusions}
This paper has proposed a novel algorithm to achieve average consensus with privacy preservation over strongly connected and balanced graphs in both discrete- and continuous-time cases. By adding edge-based perturbation signals to transmitted information, the privacy is guaranteed not to be disclosed to either internal honest-but-curious agents or external eavesdroppers. Our algorithm can guarantee accurate average consensus and has a better performance on privacy preservation than the existing related works. The proposed privacy-preserving idea can be expected to be applicable to the distributed optimization problem, which needs future study.

\bibliographystyle{ieeetr}
\bibliography{myreferences}

\end{document}